\documentclass[runningheads]{llncs}
\usepackage{amsmath}
\usepackage{amsfonts}
\usepackage{amssymb}
\usepackage{graphicx}

\newcommand{\eps}{\varepsilon}
\renewcommand{\epsilon}{\varepsilon}

\begin{document}
\title{The state complexity of star-complement-star}
\titlerunning{The state complexity of star-complement-star}
\author{
Galina Jir\'{a}skov\'{a}\,\inst{1,}%
\thanks{Research supported by  VEGA grant  2/0183/11.}
  \and 
   Jeffrey Shallit\,\inst{2}%
}
\institute{
Mathematical Institute, Slovak Academy of Sciences\\
Gre{\v s}{\' a}kova 6, 040 01 Ko\v{s}ice, Slovakia\\
\email{jiraskov@saske.sk}
  \and 
School of Computer Science, University of Waterloo \\
Waterloo, ON  N2L 3G1 Canada \\ 
\email{shallit@cs.uwaterloo.ca}
}

\maketitle

\begin{abstract}
\label{***abstract}
We resolve an open question by determining matching (asymptotic) 
upper and lower bounds on the state complexity of the
operation that sends a language $L$ to 
$\left( \, \overline{L^*} \, \right)^*$.
\end{abstract}
  
\section{Introduction}
\label{***intro}

Let $\Sigma$ be a finite nonempty alphabet, let $L \subseteq \Sigma^*$ be
a language, let $\overline{L} = \Sigma^* - L$ denote the complement of
$L$, and let $L^*$ (resp., $L^+$) denote the Kleene closure
(resp., positive closure) of the language $L$.
If $L$ is a regular language, its {\it state complexity\/}
is defined to be the number of states in the minimal deterministic
finite automaton accepting $L$ \cite{yzs94}.
In this paper we resolve an open question by determining matching
(asymptotic) upper and
lower bounds on the deterministic state complexity of the operations
\begin{eqnarray*}
L  &\rightarrow& \left( \, \overline{L^*} \, \right)^* \\
L  &\rightarrow& \left( \, \overline{L^+} \, \right)^+ .
\end{eqnarray*}

To simplify the exposition, we will write everything using an exponent notation,
using $c$ to represent complement,
as follows:
\begin{eqnarray*}
L^{+c} &:=& \overline{L^+} \\
L^{+c+} &:=& (\overline{L^+})^+,
\end{eqnarray*}
and similarly for $L^{*c}$ and $L^{*c*}$.

Note that
\begin{displaymath}
L^{*c*} = \begin{cases}
		L^{+c+},& \text{if $\epsilon \not\in L$}; \\
		L^{+c+} \ \cup \ \lbrace \epsilon \rbrace,  &
			\text{if $\epsilon \in L$}.
	\end{cases}
\end{displaymath}

It follows that the state complexity of $L^{+c+}$ and
$L^{*c*}$ differ by at most $1$.  In what follows, we will work
only with $L^{+c+}$.

\section{Upper Bound}
\label{***upper}

Consider a deterministic finite automaton
(DFA) $D=(Q_n,\Sigma,\delta,0,F)$
accepting a language $L$, where $Q_n := \{ 0,1,\ldots, n-1 \}$.
As an example, consider the three-state DFA over $\{a,b,c,d\}$
shown in Fig.~\ref{fig:d_n1} (left).
To get a nondeterministic finite automaton
(NFA) $N_1$ for the language $L^+$ from the DFA $D$,
we add an $\eps$-transition 
from every non-initial final state to the state $0$.
In our example, we add an $\eps$-transition from state $1$ to state $0$;
see Fig.~\ref{fig:d_n1} (right).
After applying the subset construction to the NFA $N_1$
we get a DFA $D_1$ for the language $L^+$.
The state set of $D_1$ consists 
of subsets of $Q_n$
see Fig.~\ref{fig:d1d2} (left). Here the sets in the labels of states
are written without commas and brackets;
thus, for example $012$ stands for the set $\{0,1,2\}$.
Next, we interchange the roles of the
final and non-final states of the DFA $D_1$,
and get a DFA $D_2$ for the language $L^{+c}$; see Fig.~\ref{fig:d1d2} (right).

To get an NFA $N_3$ for $L^{+c+}$ from the DFA $D_2$,
we add an $\eps$-transition from each non-initial final state of $D_2$
to the state $\{0\}$, see Fig.~\ref{fig:n2d3d3min} (top).
Applying the subset construction to the NFA $N_3$
results in a DFA $D_3$ for the language $L^{+c+}$
with its state set
consisting of some sets of subsets of $Q_n$;
see Fig.~\ref{fig:n2d3d3min} (middle). Here, for example, the label $0,2$
corresponds to the set $\{\{0\},\{2\}\}$.
This gives an  upper bound of $2^{2^n}$
on the state complexity of the operation plus-complement-plus.

Our first result shows that 
in the minimal DFA for $L^{+c+}$
we do not have
any state $\{S_1,S_2,\ldots,S_k\}$,
in which a set $S_i$ is a subset of some other set $S_j$;
see Fig.~\ref{fig:n2d3d3min} (bottom).
This reduces the upper bound 
to the  number of antichains of subsets
of an $n$-element set known as the Dedekind number $M(n)$ with
\cite{KM75}
$$
   {n \choose \lfloor n/2\rfloor }\le \log M(n) 
   \le {n \choose \lfloor n/2\rfloor}\Big(1+O(\frac{\log n}{n}) \Big).
$$

\vskip-10pt
\begin{figure}\label{-----fig1}
\centering
\includegraphics[scale=0.35]{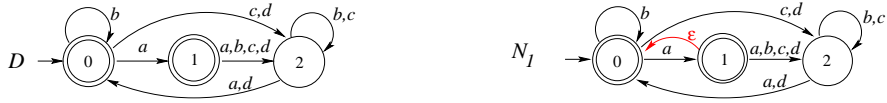}
\caption{DFA $D$ for a language $L$ and NFA $N_1$ for the language $L^+$.}
\label{fig:d_n1}
\end{figure}

\vskip-20pt
\begin{figure}[h!]\label{-----fig2}
\centering
\includegraphics[scale=0.35]{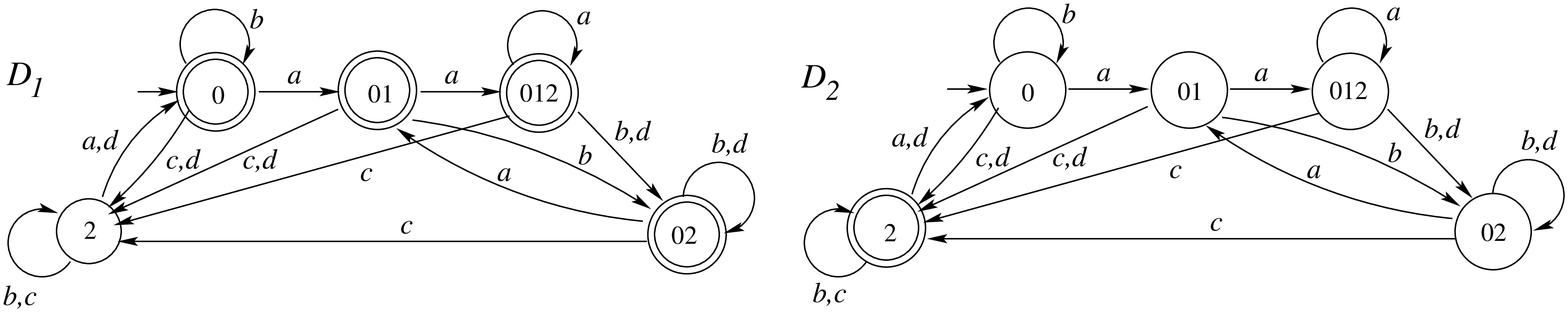}
\caption{DFA $D_1$ for  language $L^+$ and DFA $D_2$ for the language $L^{+c}$.}
\label{fig:d1d2}
\end{figure}

\begin{figure}[h!]\label{-----fig3}
\centering
\includegraphics[scale=0.40]{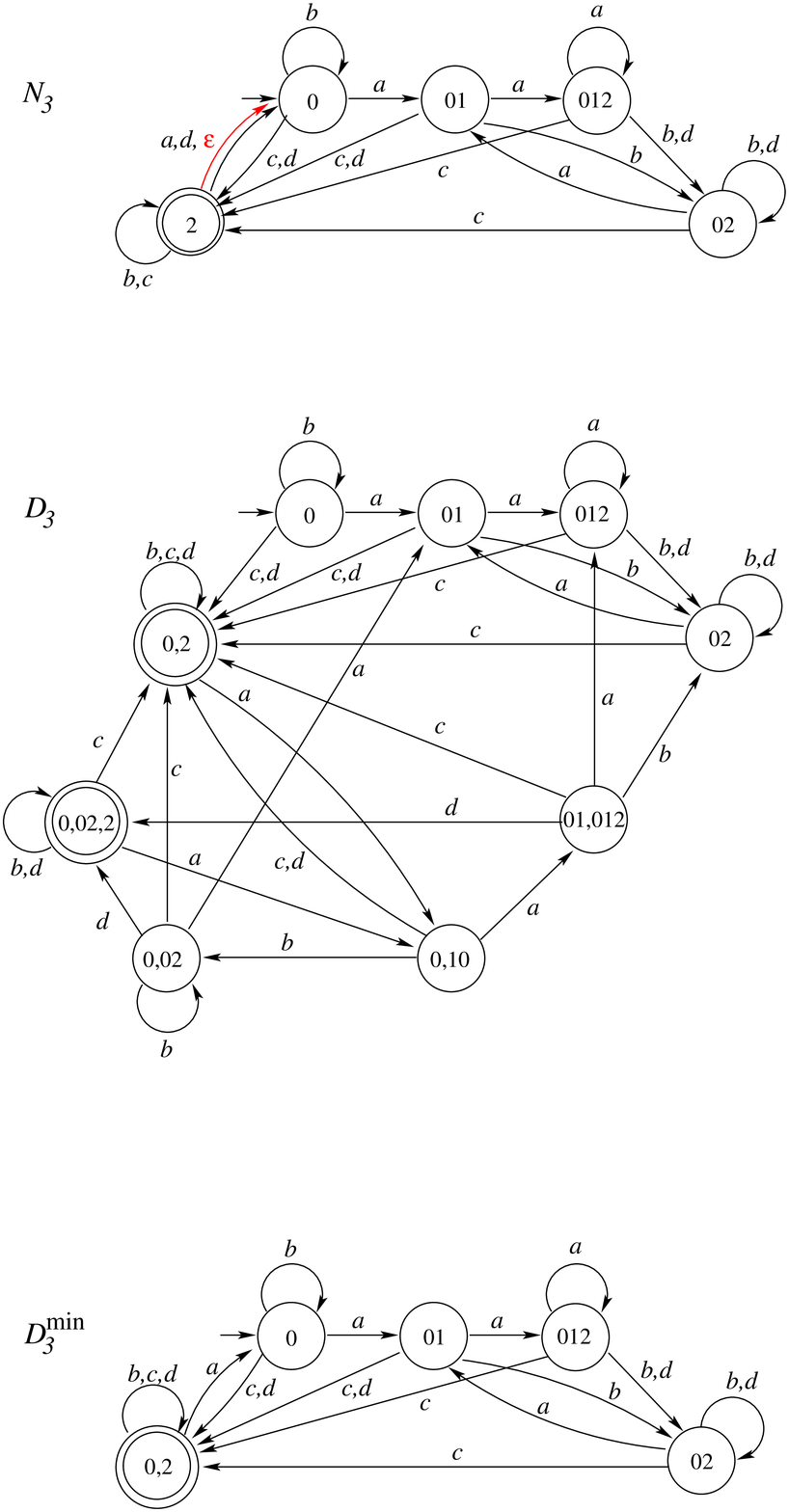}
\caption{NFA $N_3$, DFA $D_3$, and the minimal DFA $D_3^{\rm min}$
          for the language $L^{+c+}$.}
\label{fig:n2d3d3min}
\end{figure}

\begin{lemma}\label{-----le1}
 If $S$ and $T$ are subsets of $Q_n$
 such that $S\subseteq T$,
 then the states $\{S,T\}$ and $\{S\}$ of
 the DFA $D_3$ for the language $L^{+c+}$
 are equivalent.
\end{lemma}

\begin{proof}
 Let $S$ and $T$ be subsets of $Q_n$
 such that $S\subseteq T$.
 We only need to show that
 if a string $w$ is accepted by the NFA $N_3$ starting from the state $T$,
 then it also is accepted by $N_3$ from the state $S$.

 Assume  $w$ is accepted  by $N_3$ from $T$.
 Then in the NFA $N_3$, an accepting computation on $w$ from state $T$
 looks like this:
 $$
     T \stackrel{u }{\rightarrow} T_1 \stackrel{\eps }{\rightarrow}\{0\}
       \stackrel{v }{\rightarrow} T_2,
 $$
 where $w=uv$, and
 state $T$ goes to an accepting state $T_1$ on $u$ 
 without using any $\eps$-transitions,
 then $T_1$ goes to $\{0\}$ on $\eps$,
 and then $\{0\}$ goes to an accepting state $T_2$
 on $v$; it also may happen that $w=u$, in which case
 the computation ends in $T_1$.
 Let us show that $S$ goes to an accepting state of the NFA $N_3$ on $u$.

 Since $T$ goes to an accepting state $T_1$ on $u$ in the NFA $N_3$
 without using any $\eps$-transition,
 state $T$ goes to the accepting state $T_1$ in the DFA $D_2$,
 and therefore to the rejecting state $T_1$ of the DFA $D_1$.
 Thus, every state $q$ in $T$ goes to  rejecting states in the NFA $N_1$.
 Since $S\subseteq T$, every state in $S$ goes to
 rejecting states in the NFA $N_1$,
 and therefore $S$ goes to a rejecting state $S_1$ in the DFA $D_1$,
 thus to the accepting state $S_1$ in the DFA  $D_2$.
 Hence $w=uv$ is accepted from $S$ in the NFA $N_3$
 by computation
 $$
     S \stackrel{u }{\rightarrow} S_1 \stackrel{\eps }{\rightarrow}\{0\}
       \stackrel{v }{\rightarrow} T_2.
 $$
\qed
\end{proof}

Hence whenever a state $\mathcal{S}=\big\{ S_1,S_2,\ldots,S_k\}$ of the DFA $D_3$
contains two subsets $S_i$ and $S_j$ with $i\neq j$ and $S_i\subseteq S_j$,
then it is equivalet to state $\mathcal{S}\setminus \{S_j\}$.
Using this property, we get the following result.

\begin{lemma}\label{-----le2}\label{le:express}
 Let $D$ be a DFA for a language $L$ with state set $Q_n$,
 and $D_3^{\rm min}$ be the minimal DFA for $L^{+c+}$ as described above.
 Then every  state of $D_3^{\rm min}$
 can be expressed in the form
\begin{align}
 \mathcal{S}=\{X_1,X_2,\ldots,X_k\}
\label{eq1}
\end{align}
where
 \begin{itemize}
 \item  $1\le k\le n$;
\item there exist subsets 
$S_1\subseteq S_2\subseteq \cdots  \subseteq S_k\subseteq Q_n$; and
\item there exist $q_1,\ldots,q_k$, pairwise distinct states of $D$
            not in $S_k$; such that
\item $X_i=\{q_i\}\cup S_i$  for $i=1,2,\ldots,k$.
\end{itemize}
\end{lemma}

\begin{proof}
 Let $D=(Q_n,\Sigma,\delta,0,F)$. 

 For a state $q$ in $Q_n$ and a symbol $a$ in $\Sigma$, 
 let $q.a$ denote the state in $Q_n$, to~which $q$ goes on $a$,
 that is, $q.a=\delta(q,a)$.
 For a subset $X$ of $Q_n$ let $X.a$ denote
 the set of states to which states in $X$ go by $a$,
 that is,
 $$
   X.a=\bigcup_{q\in X}\{\delta(q,a)\}.
 $$
 Consider transitions on a symbol $a$
 in automata $D,N_1,D_1,D_2,N_3$; 
 Fig.~\ref{fig:transitions} illustrates these transitions.
 In the NFA $N_1$,
 each state $q$ goes to a state in $\{0,q.a\}$ if $q.a$ is a final state of $D$,
 and to state  $q.a$ if $q.a$ is non-final.
 It follows that in the DFA $D_1$ for $L^+$,
 each state $X$ (a subset of $Q_n$)
 goes on $a$ to  final state $\{0\}\cup X.a$ if $X.a$ contains a final state of $D$,
 and to non-final state $X.a$ if all states in $X.a$ are non-final in $D$.
 Hence in the DFA $D_2$ for $L^{+c}$,
 each state $X$ goes on $a$ 
 to non-final state $\{0\}\cup X.a$ if $X.a$ contains a final state of $D$,
 and to the final state $X.a$ if all states in $X.a$ are non-final in $D$.

 Therefore, in the NFA $N_3$ for $L^{+c+}$,
 each state $X$ goes on $a$ to a state in
 $\{ \{0\}, X.a\}$  if all states in $X.a$ are non-final in $D$,
 and to state  $\{0\}\cup X.a$ if $X.a$ contains a final state of $D$.

 To prove the lemma for each state,
 we use induction on the length of the shortest
 path from the initial state to the state of $D_3^{\rm min}$ in question.
 The base case is a path of length $0$.  In this case,
 the initial state is $\{\{0\}\}$, 
 which is in the required form (\ref{eq1}) with
 $k=1,q_1=0,$ and $S_1=\emptyset$.

 \begin{figure}[h]\label{-----fig4}
 \centering
 \includegraphics[scale=0.35]{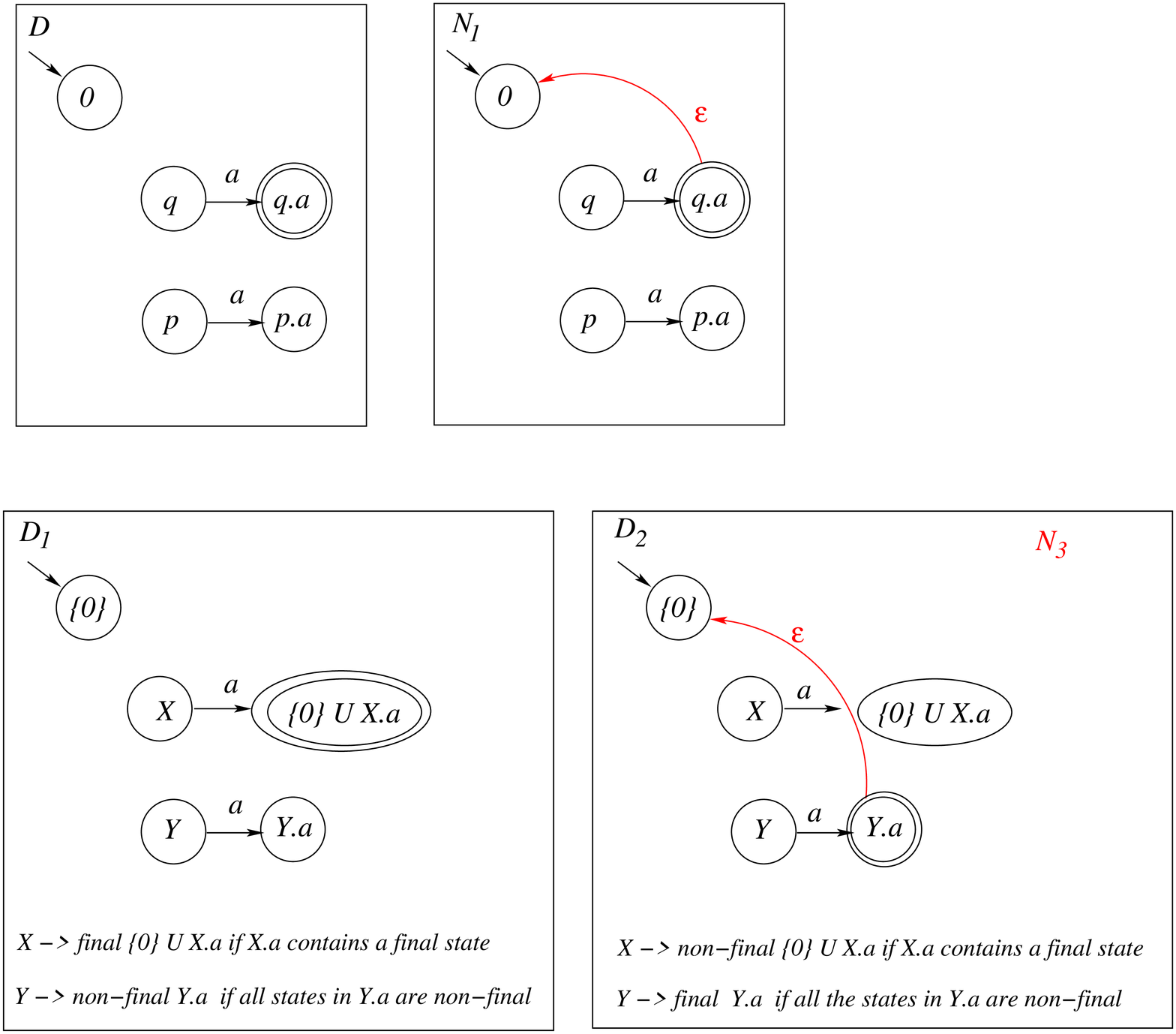}
 \caption{Transitions under symbol $a$ in automata $D, N_1,D_1,D_2,N_3$.}
 \label{fig:transitions}
 \end{figure}

 For the induction step, let
 $$
  \mathcal{S}=\{X_1,X_2,\ldots,X_k\},
 $$ 
 where $1\le k\le n$, and

 $\bullet$  $S_1\subseteq S_2\subseteq \cdots  \subseteq S_k\subseteq Q_n$,

 $\bullet$  $q_1,\ldots,q_k$ are pairwise distinct states of $D$
            that are not in $S_k$ and

 $\bullet$ $X_i=\{q_i\}\cup S_i$  for $i=1,2,\ldots,k$.

\bigskip
 We now prove the result for all states reachable from $\mathcal{S}$
 on a symbol $a$.

 First, consider the case that each $X_i$ goes on $a$
 to a non-final state $X_i'$ in the NFA~$N_3$.
 It follows that $\mathcal{S}$ goes on $a$
 to $\mathcal{S'}=\{X_1',X_2',\ldots,X_k'\}$,
 where
 $$
    X_i'= \{q_i.a\} \cup S_i.a \cup \{0\}.
 $$
 Write $p_i=q_i.a$ and $P_i=S_i.a\cup\{0\}$.
 Then we have $P_1\subseteq P_2\subseteq \cdots  \subseteq P_k\subseteq Q_n$.

 If $p_i=p_j$ for some $i,j$ with $i<j$,
 then $X_i'\subseteq X_j'$,
 and therefore $X_j'$ can be removed from state $\mathcal{S'}$ 
 in the minimal DFA $D_3^{\rm min}$.
 After  several such removals, we arrive at an equivalent state
 $$
   \mathcal{S''} = \{ X_1'',X_2'',\ldots,X_\ell''\}
 $$
 where $\ell\le k$,
 $X_i''=\{r_i\}\cup R_i$
 and the states $r_1,r_2,\ldots,r_\ell$
 are pairwise distinct.

 If $r_i\in R_\ell$ for some $i$ with $i<\ell$,
 then $X_i\subseteq R_\ell$; thus $R_\ell$ can be removed.
 After all such removals, we get an equivalent set
 $$
    \mathcal{S'''} = \{ X_1''',X_2''',\ldots,X_m'''\}
 $$ 
 where $m\le \ell$,
 $X_i'''=\{t_i\}\cup T_i$
 and the states $t_1,t_2,\ldots,t_m$
 are pairwise distinct and $t_1,t_2,\ldots,t_{m-1}$ are not in $T_m$.
 If $t_m\notin T_m$, then the state $\mathcal{S'''}$  
 is in the required form (\ref{eq1}).
 Otherwise, if $T_{m-1}$ is a proper subset of $T_m$,
 then there is a state $t$ in $T_m-T_{m-1}$,
 and then we can take $X_m'''=\{t\}\cup T_m-\{t\}$:
 since $t_1,\ldots,t_{m-1} $ are not in $T_m$,
 they are distinct from $t$,
 and moreover $T_{m-1}\subseteq T_m-\{t\}$.

 If $T_{m-1}=T_m$, then $X_{m-1}'''\supseteq X_m'''$,
 and therefore $X_{m-1}'''$ can be removed from $\mathcal{S'''}$.
 After all these removals we
 either reach some $T_i$ that is a proper subset of $T_m$,
 and then pick a state $t$ in $T_m-T_i$ in the same way as above,
 or we only get a single set $T_m$,
 which is in the required form $\{r_m\}\cup T_m-\{r_m\}$.

 This proves that if  each $X_i$ in $\mathcal{S}$ goes on $a$
 to a non-final state $X_i'$ in the NFA~$N_3$,
 then $\mathcal{S}$ goes on $a$ in the DFA $D_3^{\rm min}$ 
 to a set that is in the required form (\ref{eq1}).

 \bigskip
 Now consider the case that 
 at least one  $X_j$ in $\mathcal{S}$
 goes to a final state $X_j'$ in the NFA $N_3.$
 It follows that 
 $\mathcal{S}$ goes to a final state
 $$
   \mathcal{S'}=\{ \{0\}, X_1',X_2',\ldots,X_k'\},
 $$
 where $X_j'=\{q_j.a\}\cup S_j.a$ and if $i\neq j$,
 then $X_i'=\{q_i.a\}\cup S_i.a$ or $X_i'=\{0\}\cup\{q_i.a\}\cup S_i.a$
 We now can remove all $X_i$ that contain state $0$,
 and arrive at an equivalent state 
 $$
   \mathcal{S''}=\{ \{0\}, X_1'',X_2'',\ldots,X_\ell''\},
 $$
 where $\ell\le k$, and
 $X_i''= \{p_i\}\cup P_i$, and
 $P_1\subseteq P_2 \subseteq \cdots \subseteq P_\ell \subseteq Q_n $,
 and each $p_i$ is distinct from $0$.

 Now in the same way as above we arrive at an equivalent state 
 $$
  \{ \{0\}, \{t_1\}\cup T_1, \ldots, \{t_m\}\cup T_m \}
 $$
 where $m\le \ell$,
 all the $t_i$ are pairwise distinct and different from $0$,
 and moreover, the states $t_1,\ldots,t_{m-1}$ are not in $T_m$.
 If $t_m$ is not in $T_m$, then we are done.
 Otherwise, we remove all sets with $T_i=T_m$.
 We either arrive at a proper subset $T_j$ of $T_m$,
 and may pick a state $t$ in $T_m-T_j$ to play the role of new $t_m$,
 or we arrive at $\{\{0\}, T_m\}$,
 which is in the required form $\{\{0\}\cup\emptyset, t_m\cup T_m-\{t_m\}\}$.
 This completes the proof of the lemma.
\qed
\end{proof}

\begin{corollary}[Star-Complement-Star: Upper Bound]\label{-----co1}
 If a language $L$ is accepted by a DFA of $n$ states,
 then the language $L^{*c*}$ is accepted by a DFA of $2^{O(n\log n)}$ states.
\end{corollary}

\begin{proof}
 Lemma~\ref{le:express} gives the following upper bound
 $$
    \sum_{k=1}^{n}   {n \choose k} k! (k+1)^{n-k}
 $$
 since 
 we first choose any permutation of $k$ distinct elements $q_1,\ldots,q_k$,
 and then
 represent each set $S_i$ as disjoint union of sets $S_1', S_2',\ldots, S_i'$
 given by a function $f$ from $Q_n-\{q_1,\ldots,q_k\}$ to $\{1,2,\ldots,k+1\}$
 as follows:
 $$
     S_i'=\{q\mid f(q)=i\}, \qquad  
            S_i=S_1'\ \dot\cup \ S_2'\ \dot\cup \ \cdots \ \dot\cup \ S_i',
 $$
 while states with $f(q)=k+1$ will be outside each $S_i'$;
 here $\dot\cup$ denotes a disjoint union.
 Next, we  have
 $$
    \sum_{k=1}^{n}   {n \choose k} k! (k+1)^{n-k}\le
    n! \sum_{k=1}^{n}   {n \choose k}(n+1)^{n-k} \le
   n!(n+2)^n= 2^{O(n\log n)},
 $$
 and the upper bound follows.
\qed
\end{proof}

\begin{remark}
The summation $\sum_{k=1}^n {n \choose k} k! (k+1)^{n-k}$ differs by
one from Sloane's sequence A072597 \cite{Sloane}.  These numbers are the
coefficients of the exponential generating function of
$1/(e^{-x} - x)$.  It follows, by standard techniques, that
these numbers are asymptotically given by
$C_1 W(1)^{-n} n!$,
where
$$W(1) \doteq .5671432904097838729999686622103555497538$$ is the
Lambert W-function evaluated at $1$, equal to the positive real
solution of the equation $e^x = 1/x$, and $C_1$ is a constant,
approximately 
$$1.12511909098678593170279439143182676599.$$
The convergence
is quite fast; this gives a somewhat more explicit version of the
upper bound.
\end{remark}

\section{Lower Bound}
\label{***lower}

We now turn to the matching lower bound on the state complexity
of plus-complement-plus.  The basic idea is to create one DFA
where the DFA for $L^{+c+}$ has many reachable states, and another
where the DFA for $L^{+c+}$ has many distinguishable states.
Then we ``join'' them together in Corollary~\ref{-----co2}.

The following lemma uses a four-letter alphabet 
to prove the reachability of some specific states
of the DFA $D_3$ for plus-complement-plus.

\begin{lemma}\label{-----le3}\label{le:reach}
 There exists an $n$-state
 DFA  $D=(Q_n,\{a,b,c,d\},\delta,0,\{0,1\})$
 such that in the  DFA $D_3$ for the language $L(D)^{+c+}$
 every state of the form
 $$
\Big\{  \{0,q_1\}\cup S_1,  \{0,q_2\}\cup S_2,\ldots, \{0,q_k\}\cup S_k\Big\} 
 $$
 is reachable, 
 where $1\le k \le n-2$, 
 $S_1,S_2,\ldots,S_k$ are subsets of $\{2,3,\ldots,n-2\}$ with
 $S_1\subseteq S_2\subseteq \cdots  \subseteq S_k$, and the
 $q_1,\ldots,q_k$ are pairwise distinct states in $\{2,3,\ldots,n-2\}$
 that are not in $S_k$.
\end{lemma}

\begin{proof}
 Consider the DFA $D$ over $\{a,b,c,d\}$
 shown in  Fig.~\ref{fig:max_reach}.
 Let $L$ be the language accepted by the DFA $D$.
 
 \begin{figure}[t]\label{-----fig5}
 \centering
 \includegraphics[scale=0.40]{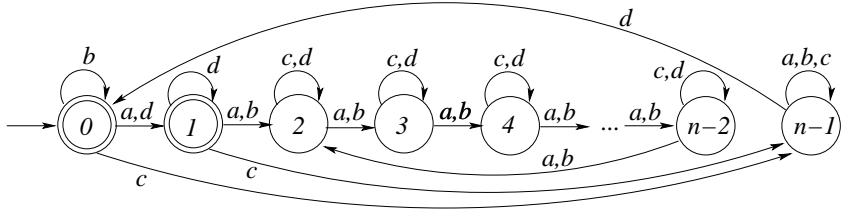}
 \caption{DFA $D$ over $\{a,b,c,d\}$ with many reachable states in 
         DFA $D_3$ for $L^{+c+}$.}
 \label{fig:max_reach}
 \end{figure}

 Construct the NFA $N_1$ for the language $L^+$ from the DFA $D$
 by adding loops on $a$ and $d$ in the initial state $0$.
 In the subset automaton corresponding to the NFA $N_1$,
 every subset of $\{0,1,\ldots,n-2\}$ containing state $0$
 is reachable from the initial state $\{0\}$ on a string over $\{a,b\}$
 since each subset $\{0,i_1,i_2,\ldots, i_k\}$ of size $k$,
 where $1\le k\le n-1$ and $1\le i_1<i_2<\cdots<i_k\le n-2$,
 is reached from the set $\{0,i_2-i_1,\ldots,i_k-i_1\}$ of size $k-1$
 on the string $ab^{i_1-1}$.
 Moreover, after reading every symbol of string $ab^{i_1-1}$,
 the subset automaton is always in a set that contains state $0$.
 All such states are rejecting in the DFA $D_2$ for the language $L^{+c}$,
 and therefore, in the NFA $N_3$ for $L^{+c+}$,
 the initial state $\{0\}$ only goes to the rejecting state
 $\{0,i_1,i_2,\ldots, i_k\}$ on $ab^{i_1-1}$.

 Hence in the DFA $D_3$, for every subset $S$  of $\{0,1,\ldots,n-2\}$
 containing  $0$,
 the initial state $\{\{0\}\}$ goes to the state $\{S\}$
 on a string $w$ over $\{a,b\}$.

 Now notice that transitions on symbols $a$ and $b$
 perform the cyclic permutation of states  in $\{2,3,\ldots,n-2\}$.
 For every state $q$ in $\{2,3,\ldots,n-2\}$ and an integer $i$,
 let 
 $$
    q\ominus i = ((q-i-2) \bmod n-3) + 2
 $$
 denote the state in $\{2,3,\ldots,n-2\}$
 that goes to the state $q$ on string $a^i$,
 and, in fact, on every string over $\{a,b\}$ of length $i$.
 Next, for a subset $S$ of $\{2,3,\ldots,n-2\}$ let
 $$
    S\ominus i=\{q\ominus i\mid q\in S\}.
 $$
 Thus $S\ominus i$ is a shift of $S$,
 and if $q\notin S$, then $q\ominus i\notin S\ominus i$.

 The proof of the lemma now proceeds by induction on $k$.
 To prove the base case, let $S_1$ be a subset of $\{2,3,\ldots,n-2\}$
 and $q_1$ be a state in $\{2,3,\ldots,n-2\}$ with $q_1\notin S_1$.
 In the NFA $N_3$, the initial state $\{0\}$ goes to the state $\{0\}\cup S_1$
 on a string $w$ over $\{a,b\}$.
 Next, state $q_1\ominus|w|$ is  in $\{2,3,\ldots,n-2\}$,
 and it is reached from state $1$ on a string $b^\ell$,
 while state $0$ goes to itself on $b$.
 In the DFA $D_3$ we thus have
 $$
   \big\{ \{0\} \big\}\stackrel{a}{\rightarrow}\big\{ \{0,1\} \big\}
         \stackrel{b^\ell}{\rightarrow}        \big\{ \{0,q_1\ominus|w|\} \big\}
         \stackrel{w}{\rightarrow}             \big\{ \{0,q_1\}\cup S_1 \big\},
 $$
 which proves the base case.

 Now assume that every set of size $k-1$ satisfying the lemma
 is reachable in the DFA $D_3$.
 Let 
 $$
  \mathcal{S}=\Big\{\{0,q_1\}\cup S_1,\{0,q_2\}\cup S_2,\ldots,\{0,q_k\}\cup S_k\Big\} 
 $$
 be a set of size $k$ satisfying the lemma.
 Let $w$ be a string, on which $\big\{ \{0\} \big\}$
 goes to $\big\{ \{0\}\cup S_1 \big\}$,
 and let $\ell$ be an integer such that  $1$ goes to $q_1\ominus|w|$ on $b^\ell$.
 Let
 $$
 \mathcal{S'}=\Big\{\{0,q_2\ominus|w|\ominus\ell\}   \cup 
                        S_2\ominus|w|\ominus\ell,    \ldots,
                    \{0,q_k\ominus|w|\ominus\ell\}    \cup 
                        S_k\ominus|w|\ominus\ell      \Big\},
 $$
 where the operation $\ominus$ is understood to have left-associativity.
 Then $\mathcal{S'}$ is reachable by induction. 
 On $c$,
 every set $\{0,q_i\ominus|w|\ominus\ell\}\cup S_i\ominus|w|\ominus\ell$
 goes  to the accepting state 
 $\{n-1,q_i\ominus|w|\ominus\ell \}\cup S_i\ominus|w|\ominus\ell$
 in the NFA $N_3$, 
 and therefore also to the initial state $\{0\}$.
 Then, on $d$, every state 
  $\{n-1,q_i\ominus|w|\ominus\ell\}\cup S_i\ominus |w|\ominus\ell$
 goes to the rejecting state 
 $\{0,q_i\ominus|w|\ominus\ell \}\cup S_i\ominus|w|\ominus\ell $,
 while $\{0\}$ goes to $\{0,1\}$.
 Hence, in the DFA $D_3$ we have 
 \begin{align*}
   \mathcal{S'}
  &\stackrel{c}{\rightarrow} 
    \Big\{\{0\},
     \{n-1,q_2\ominus|w|\ominus\ell\}\cup S_2\ominus|w|\ominus\ell,  \ldots,
     \{n-1,q_k\ominus|w|\ominus\ell\}\cup S_k\ominus|w|\ominus\ell\Big\} \\
  & \stackrel{d}{\rightarrow}
    \Big\{\{0,1\}, \{0,q_2\ominus|w|\ominus\ell\}\cup S_2\ominus|w|\ominus\ell,
     \ldots, \{0,q_k\ominus|w|\ominus\ell\}\cup S_k\ominus|w|\ominus\ell\Big\} \\
  & \stackrel{b^\ell}{\rightarrow}
   \Big\{\{0,q_1\ominus |w|\}, \{0,q_2\ominus |w|\}\cup S_2\ominus |w|,\ldots,
                                \{0,q_k\ominus |w|\}\cup S_k\ominus |w|\Big\}
     \stackrel{w}{\rightarrow}\mathcal{S}.
 \end{align*}
 It follows that $\mathcal{S}$ is reachable in the DFA $D_3$.
 This  concludes the proof.
\qed
\end{proof}

The next lemma 
shows that some  rejecting states of the DFA $D_3$, 
in which no set is a subset of some other set,
may be pairwise distinguishable.
To prove the result
it uses four symbols,
one of which is the symbol $b$ from the proof of the previuos lemma.

\begin{lemma}\label{-----le4}\label{le:equiv}
 Let $n\ge5$.
 There exists an $n$-state DFA $D=(Q_n,\Sigma,\delta,0,\{0,1\})$
 over a four-letter  alphabet $\Sigma$
 such that all the states of the DFA $D_3$ for the language $L(D)^{+c+}$
 of the form
 $$
    \Big\{ \{0\}\cup T_1,  \{0\}\cup T_2,\ldots,  \{0\}\cup T_k \Big\},
 $$
 in which no set is a subset of some other set and 
 each $T_i\subseteq\{2,3,\ldots,n-2\}$,
 are~pairwise distinguishable.
\end{lemma} 

\begin{proof}
 To prove the lemma, we reuse the symbol $b$
 from the proof of Lemma~\ref{le:reach},
 and define three new symbols $e,f,g$
 as shown in Fig.~\ref{fig:befg}.

 \begin{figure}[t]\label{-----fig6}
 \centering
 \includegraphics[scale=0.40]{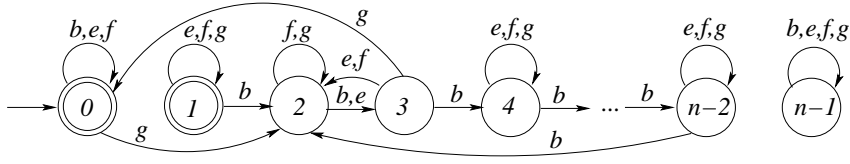}
 \caption{DFA $D$ over $\{b,e,f,g\}$ with many distinguishable states in 
         DFA $D_3$.}
 \label{fig:befg}
 \end{figure}
 
 Notice that on states $2,3,\ldots,n-2$,
 the symbol $b$ performs a big permutation,
 while $e$ performs a trasposition,
 and $f$  a contraction.
 It follows that every transformation of states $2,3,\ldots,n-2$
 can be performed by strings over $\{b,e,f\}$.
 In particular, for each subset $T$ of $\{2,3,\ldots,n-2\}$,
 there is a string $w_T$ over $\{b,e,f\}$
 such that in  $D$,
 each state in $T$ goes to state $2$ on $w_T$,
 while each state in $\{2,3,\ldots,n-2\}\setminus T$ 
 goes to state $3$ on $w_T$.
 Moreover, state $0$ remains in itself while reading the string~$w_T$.
 Next, the symbol $g$ sends state $0$ to state $2$,
 state $3$ to state $0$,
 and state $2$ to itself.

 It follows that in the NFA $N_3$,
 the  state $\{0\}\cup T$,
 as well as each state $\{0\}\cup T'$ with $T'\subseteq T$,
 goes to the accepting state $\{2\}$ on $w_T\cdot g$.
 However,
 every other state $\{0\}\cup T''$ with $T''\subseteq\{2,3,\ldots,n-2\}$
 is  in a state containig $0$,
 thus in a rejecting state of $N_3$,
 while reading $w_T\cdot g$, 
 and it is in the rejecting state $\{0,3\}$
 after reading $w_T$.
 Then $\{0,3\}$ goes to the rejecting state $\{0,2\}$ on reading $g$.
 
 Hence the string $w_T\cdot g$ is accepted by the NFA $N_3$
 from each state $\{0\}\cup T'$ with $T'\subseteq T$,
 but rejected from any other state $\{0\}\cup T''$
 with $T''\subseteq\{2,3,\ldots,n-2\}$.
 
 Now consider  two different states of the DFA $D_3$
 \begin{align*}
    & \mathcal{T}=\big\{  \{0\} \cup T_1, \ldots, \{0\}\cup T_k \big\},\\
    & \mathcal{R}=\big\{  \{0\} \cup R_1, \ldots, \{0\}\cup R_\ell \big\}, 
 \end{align*}
 in which no set is a subset of some other set and where
 each $T_i$ and each $R_j$ is a subset of $\{2,3,\ldots,n-2\}$.
 Then, without loss of generality,
 there is a set $\{0\}\cup T_i$ in $\mathcal{T}$ that is not in $\mathcal{R}$.
 If no set $\{0\}\cup T'$ with $T'\subseteq T_i$ is in $\mathcal{R}$,
 then the string $w_{T_i}\cdot g$
 is accepted from $\mathcal{T}$ but not from $\mathcal{R}$.
 If there is a subset $T'$ of $T_i$
 such that $\{0\}\cup T'$ is  in $\mathcal{R}$,
    then for each suset $T''$ of $T'$
    the set $\{0\}\cup T''$  cannot be in $\mathcal{T}$,
    and then the string  $w_{T'}\cdot g$
 is accepted from $\mathcal{R}$ but not from $\mathcal{T}$.
\qed
\end{proof}

\begin{corollary}[Star-Complement-Star: Lower Bound]\label{-----co2}
 There exists a language $L$ 
 accepted by an $n$-state DFA  over a seven-letter input alphabet,
 such that any DFA for the language $L^{*c*}$ 
 has  $2^{\Omega(n\log n)}$ states.
\end{corollary}

\begin{proof}
 Let 
 $\Sigma=\{a,b,c,d,e,f,g\}$ and
 $L$ be the language 
 accepted by $n$-state DFA $D=(\{0,1,\ldots,n-1\},\Sigma,\delta,0,\{0,1\})$,
 where transitions on symbols $a,b,c,d$ 
 are defined as in the proof of Lemma~\ref{le:reach},
 and on symbols $d,e,f$ as in the proof of Lemma~\ref{le:equiv}.

 Let $m=\lceil n/2  \rceil$.
 By Lemma~\ref{le:reach},
 the following states
 are reachable in the DFA $D_3$    
 for $L^{+c+}$:
 $$
     \{ \{0,2\}\cup S_1, \{0,3\}\cup S_2,\ldots, 
                   \{0,m-2\}\cup S_{m-1}\},
 $$
 where 
 $S_1\subseteq S_2 \subseteq\cdots \subseteq S_{m-1}\subseteq\{m-1,m,\ldots,n-2\}$.
 The number of such subsets $S_i$ is given by $m^{n-m}$, and we have
 $$
  m^{n-m}\ge \Big( \frac{n}{2}\Big) ^{\frac{n}{2}-1}=2^{\Omega(n\log n)}.
 $$
 By Lemma~\ref{le:equiv}, all these states are pairwise distinguishable,
 and the lower bound follows. 
\qed
\end{proof}

Hence we have an asymptotically tight bound on the state complexity of
star-complement-star operation that is significantly smaller than
$2^{2^n}$.

\begin{theorem}\label{-----thm1}
 The state complexity of star-complement-star is $2^{\Theta(n \log n)}$.
\qed
\end{theorem}

\section{Applications}
\label{***applications}

We conclude with an application.

\begin{corollary}
Let $L$ be a regular language, accepted by a DFA with $n$ states.  Then
any language that can be expressed in terms of $L$ and the operations
of positive closure, Kleene closure, and complement has state complexity
bounded by $2^{\Theta(n \log n)}$.
\end{corollary}

\begin{proof}
As shown in \cite{BGS}, every such language
can be expressed, up to inclusion of~$\epsilon$,
as one of the following $5$ languages and their complements:
$$L, L^{+}, L^{c+}, L^{+c+}, L^{c+c+}.$$
If the state complexity of $L$ is $n$, then
clearly the state complexity of $L^c$ is also~$n$.  Furthermore,
we know that the state complexity of $L^+$ is bounded by
$2^n$ (a more exact bound can be found in \cite{yzs94}); this
also handles $L^{c+}$.
The remaining languages can be handled with Theorem~\ref{-----thm1}.
\qed
\end{proof}

\end{document}